\documentclass{amsart}
\usepackage{amssymb}

\usepackage{graphicx}
\usepackage{amscd}

\newtheorem{theorem}{Theorem}
\theoremstyle{plain}

\newtheorem{lemma}{Lemma}

\newtheorem{remark}{Remark}

\numberwithin{equation}{section}

\begin{document}

\bigskip
\bigskip
\bigskip

\begin{center}
ON THE GENERALIZED HILL PROCESS FOR SMALL PARAMETERS AND APPLICATIONS
\end{center}

\bigskip

\begin{center}
Gane Samb LO $^{*}$,  El Hadji DEME $^{**}$ and Aliou DIOP $^{**}$\\

\bigskip

$^{*}$ LSTA, UPMC, France and LERSTAD, Universit\'e Gaston Berger
de Saint-Louis, SENEGAL\\
gane-samb.lo@ugb.edu.sn, ganesamblo@ufrsat.org, www.lsta.upmc.fr\\

\bigskip

${**}$ LERSTAD, Universit\'e Gaston Berger de Saint-Louis, SENEGAL\\
ehdeme@ufrsat.org, aliou.diop@ugb.edu.sn\\
\end{center}

\bigskip

\noindent \textit{ABSTRACT}.\\
Let $X_{1},X_{2},...$ be a sequence of independent copies (s.i.c) of a real
random variable (r.v.) $X\geq 1$, with distribution function $df$ $F(x)=\mathbb{P}%
(X\leq x)$ and let $X_{1,n}\leq X_{2,n} \leq ... \leq X_{n,n}$ be the order
statistics based on the $n\geq 1$ first of these observations. The following
continuous generalized Hill process 
\begin{equation*}
T_{n}(\tau )=k^{-\tau }\sum_{j=1}^{j=k}j^{\tau }\left( \log X_{n-j+1,n}-\log
X_{n-j,n}\right) ,  \label{dl02}
\end{equation*}
$\tau >0$, $1\leq k \leq n$, has been introduced as a continuous family of
estimators of the extreme value index, and largely studied for statistical
purposes with asymptotic normality results restricted to $\tau > 1/2$. We
extend those results to $0 < \tau \leq 1/2$ and show that asymptotic
normality is still valid for $\tau=1/2$. For $0 < \tau <1/2$, we get non
Gaussian asymptotic laws which are closely related to the Riemann function $%
\zeta(s)=\sum_{n=1}^{\infty} n^{-s},s>1$. \\

\bigskip

\textbf{Secondary} 60B10, 60F17.\\

\bigskip


\Large

\section{Introduction}

\label{intro}

The aim of this note is to settle an asymptotic theory for some functional
forms of Hill's estimators. Precisely, let $X_{1},X_{2},...$ be a sequence
of independent copies (s.i.c) of a real random variable (r.v.) $X$ with
distribution function ($df$) $F(x)=\mathbb{P}(X\leq x)$, $x\in\mathbb{R}$. Since we are only
concerned with the upper tail of $F$ in this paper, we assume without loss
of generality that $X\geq 1$ and define a s.i.c. of the r.v. $Y=\log X$
denoted $Y_{1},Y_{2},...$ with $df$ $G(x)=P(Y\leq x)=F(e^{x}$), $x\geq 0.$
Finally $Y_{1,n}=\log X_{1,n}\leq ...\leq Y_{n,n}=\log X_{n,n}$ are their
respective order statistics. The following generalized continuous Hill
process 
\begin{equation}
T_{n}(\tau)=k^{-\tau }\sum_{j=1}^{j=k}j^{\tau }\left( \log X_{n-j+1,n}-\log
X_{n-j,n}\right) ,  \label{dl02a}
\end{equation}
has been introduced and studied in \cite{dioplo} and \cite{dioplo90},
for a continuous parameter $\tau >0$, and as throughout this paper, 
\begin{equation*}
1\leq k=k(n)\leq n,\text{ }k\rightarrow +\infty ,\text{ }k/n\rightarrow 0%
\text{ }as\text{ }n\rightarrow +\infty .
\end{equation*}
But some margins of (\ref{dl02a}) appeared before that, since $T_{n}(1)$ is
the Hill estimator introduced in \cite{hill} in 1975. Also, De Haan and
Resnick \cite{haanresnick} proposed $DR_{n}=T_{n}(0)/(\log k)$ as a
simple estimator of the extreme value index. It has been shown in \cite{lothese}
that for $F \in D(\Lambda)$, $DR_{n}$ - when appropriately centered and
normalized - converges in law to a multiple of a Gumbel law.\newline

\bigskip

\noindent However, we are only concerned in this paper with values $\tau >0$. In this case, the process (\ref{dl02a}) provides a family of estimators of
the Univariate Extreme Value Index as shown in \cite{dioplo90} in a sense to be
defined shortly, after a brief recall on Extreme Value Theory.

\bigskip \noindent Since this work arises in the just mentioned
theory, it would be appropriate to make some reminders before going any
further. The reader is referred to de Haan \cite{dehaan}, de Haan and
Ferreira \cite{dehaan1}, Resnick \cite{resnick} and Galambos \cite
{galambos} for a modern and large account of the Extreme Value Theory.
However, the least to say is that the $df$ $F$ is said to
be attracted to a non degenerated $df$ $M$ iff the maximum $X_{n,n}=\max
\left( X_{1},...X_{n}\right)$, when appropriately centered and
normalized \ by two sequences of real numbers $\left( a_{n}>0\right) _{n\geq
0}$ and $\left( b_{n}\right) _{n\geq 0}$, converges to $M$, in the sense that

\begin{equation}
\lim_{n\rightarrow +\infty }\mathbb{P}\left( X_{n,n}\leq a_{n}\text{ }x+b_{n}\right)
=\lim_{n\rightarrow +\infty }F^{n}\left( a_{n}x+b_{n}\right) =M(x),
\label{dl05}
\end{equation}
for continuity points $x$ of $M$. If (\ref{dl05}) holds, it is said that $F$
is attracted to $M$ or $F$ belongs to the domain of attraction of $M$,
written $F\in D(M).$ It is well-kwown that the three nondegenerate possible
limits in (\ref{dl05}), called extremal $df$'s, are the following.
\newline

\noindent The Gumbel $df$ 
\begin{equation}
\Lambda (x)=\exp (-\exp (-x)),\text{ }x\in \mathbb{R},  \label{dl05a}
\end{equation}
the Fr\'{e}chet $df$ with parameter $\gamma >0$

\begin{equation}
\phi _{\gamma }(x)=\exp (-x^{-\gamma })\mathbb{I}_{\left[ 0,+\infty \right[
}(x),\text{ }x\in \mathbb{R}\   \label{dl05b}
\end{equation}
and the Weiblull $df$ with parameter $\beta >0$

\begin{equation}
\psi _{\beta}(x)=\exp (-(x)^{\beta})\mathbb{I}_{\left] -\infty ,0\right]
}(x)+(1-1_{\left] -\infty ,0\right] }(x)),\ \ x\in \mathbb{R},\ 
\label{dl05c}
\end{equation}
where $\mathbb{I}_{A}$ denotes the indicator function of the set A.

\bigskip

\noindent Actually the limiting $df$ $M$ is defined by an equivalence
class of the binary relation $\mathcal{R}$ on the set $\mathcal{D}$ of $%
df^{\prime}s$ on $\mathbb{R}$, defined as follows 
\begin{equation*}
\forall (M_{1},M_{2})\in \mathcal{D}^{2},(M_{1}\text{ }\mathcal{R}\text{ }%
M_{2})\Leftrightarrow \exists (a,b)\in \mathbb{R}_{+}\backslash \{0\}\times 
\mathbb{R},\forall (x\in \mathbb{R}) :
\end{equation*}
\begin{equation*}
M_{2}(x)=M_{1}(ax+b).
\end{equation*}
One easily checks that if $F^{n}\left( a_{n}x+b_{n}\right) \rightarrow
M_{1}(x),$ then $F^{n}\left( c_{n}x+d_{n}\right) \rightarrow
M_{1}(ax+b)=M_{2}(x)$ whenever 
\begin{equation}
a_{n}/d_{n}\rightarrow a\text{ and }(b_{n}-d_{n})/c_{n}\rightarrow b\text{
as }n\rightarrow \infty .  \label{dl05f}
\end{equation}
\noindent Theses facts allow to parameterize the class of extremal $df$'s.
For this purpose, suppose that (\ref{dl05}) holds for the three $df$'s given
in (\ref{dl05a}), (\ref{dl05b}) and (\ref{dl05c}). If we take sequences $%
(c_{n}>0)_{n\geq 1}$ and $(d_{n})_{n\geq 1}$ such that the limits in (\ref
{dl05f}) are $a=\gamma=1/\alpha$ and $b=1$ (in the case of Fr\'{e}chet
extremal domain), and $a=-\beta=-1/\alpha$ and $b=-1$ (in the case of
Weibull extremal domain), and finally, if we interpret $(1+\gamma
x)^{-1/\gamma}$ as $exp(-x)$ for $\gamma =0$ (in the case of Gumbel extremal
domain), we are entitled to write the three extremal $df$'s in the
parameterized shape 
\begin{equation}
G_{\gamma }(x)=\exp (-(1+\gamma x)^{-1/\gamma }),\text{ }1+\gamma x\geq 0,
\label{dl05d}
\end{equation}
called the Generalized Extreme Value (GEV) $df$ with parameter $\gamma \in 
\mathbb{R}$.

\bigskip

\noindent The parameter $\gamma$ is called extreme value index. Originally,
Hill(1975) \cite{hill} introduced the so-called Hill estimator $T_{n}(1)$
of the parameter $\gamma $ of the Pareto $df$ $\left\{
1-Cx^{-1/\gamma }\right\} \mathbb{I}_{(x\geq 0)}$. So did Pickands \cite
{pickands} as well. In the same way, the Generalized Hill process $T_{n}(\tau )$
provides an infinite family of estimators for the extremal index value in the sense
that for any $\tau >0,$ for F$\in D(G_{\gamma }), \text{ } 0 <\gamma \leq +\infty$,
\begin{equation*}
\tau T_{n}(\tau )\rightarrow \gamma \: \: a.s \: \: as \: \: n \rightarrow
+\infty.
\end{equation*}
This motivated Diop and Lo \cite{dioplo} to find out asymptotic normality
results for this family in order to construct statistical tests. Although
their results seem to be valuable, they restricted themselves to values 
$\tau >1/2$. Actually, this restriction is imposed by the so-called Hungarian
approximation methods based on M. Cs\H{o}rg\"{o} \textit{et al.} \cite{cchm} results that they used. It is then necessary to call for another approach to get around this obstacle.

\bigskip

\noindent In this paper, we will complete the work of \cite{dioplo} by
giving the asymptotic law of $T_{n}(\tau )$ for $0<\tau \leq 1/2$ under the hypothesis $F \in D(G_{\gamma}), \: 0 \leq \gamma < +\infty$. Our results are also complements to the related ones for $T_{n}(0)$ in connection with those in \cite{pickands}, \cite{haanresnick} and \cite {lothese}. We point out for once that $D(G_{0})=D(\Gamma)$ for $\gamma=+\infty$.\newline

\noindent We restrict ourselves ourselves here to the cases $0 \leq \gamma <+\infty$,  since the case $\gamma <0$, may be studied through the transform  $F(x_{0}(F)-1/\dot{})) \in D(G_{-\gamma})$ for estimating $\gamma$. This leads to replace $X_{n-j+1,n}$ by $x_{0}(F)-1/1/X_{j,n}$ in (\ref{dl02a}). However, a direct investigation of (\ref{dl02a}) for $\gamma <0$ is possible. This requires the theory of sums of dependent random variables while this paper uses results on sums of independant random variabes, as it will be seen shortly. We consequently  consider a special handling of this case in a dinstinct paper.\\

\bigskip

\noindent Our achievement is the full description of the asymptotic behavior of $%
T_{n}(\tau )$, under the hypothesis $F \in D(G_{\gamma}), 0 \leq \gamma < +\infty$, for $0<\tau \leq 1/2$, which is still Gaussian for $\tau =1/2$ and non Gaussian for $\tau <1/2$. In the latter case, the limiting laws are described in connection with the Riemann function. Our best results concern $\tau \in ]0,1/2[$, in which case the usual conditions, $(C2)$ is useless and and $(C3)$ becomes very week. The asymptotic results become very general (see Subsection \ref{subsec3}).
\newline

\bigskip

\noindent Before we state our results, we need to introduce some
representations used in the theorems and to describe the conditions we shall
require, at the appropriate time, in the next section. In Section 
\ref{results}, we will state our results and their proofs. In Section \ref{remarks},
statistical illustrations of the results and simulation study
results are provided while Section \ref{sectiontool} includes the key tools of this
paper, which ends by a general conclusion in Section \ref{conclusion}.

\section{Representations Tools}

\label{tools}

Throughout the paper, we use the usual representation of $Y_{1},Y_{2},...$
by $G^{-1}(1-U_{1}),G^{-1}(1-U_{2}),...$ where $U_{1},U_{2},...$ are
independent and uniform random variables on $(0,1)$ in the sense of equality
in distribution (denoted by $=_{d})$%
\begin{equation*}
\left\{ Y_{j},n\geq 1\}=_{d}\{G^{-1}(1-U_{j}),n\geq 1\right\} ,
\end{equation*}
and hence 
\begin{equation*}
\{\left\{ Y_{1,n},Y_{2,n},...Y_{n,n}\},n\geq1\ \right\}
\end{equation*}

\begin{equation}
=_{d}\left\{
\{G^{-1}(1-U_{n,n}),G^{-1}(1-U_{n-1,n}),...,G^{-1}(1-U_{1,n})\},n\geq
1\right\} .  \label{repre}
\end{equation}
In connexion with this, we also use the following Malmquist representation
(see \cite{shwell}, p. 336)

\begin{equation}
\{\log (\frac{U_{j+1,n}}{U_{j,n}})^{j},j=1,...,n\}=_{d}\{E_{1},...,E_{n}\},
\label{malm}
\end{equation}

\noindent where $E_{1},...,E_{n}$ are independent standard exponential
random variables.\newline

\bigskip

\noindent Now we recall the classical representations of $df$'s attracted to
some nondegenerated extremal $df$. For each $df$ $F$ in the extremal domain, an appropriate
representation is given for $G^{-1}(1-u)=logF^{-1}(1-u), \: u \in (0,1)$.

\begin{theorem}
\label{A} We have :

\begin{enumerate}
\item  Karamata's representation (KARARE)

\noindent (a) If $F\in D(G_{\gamma})$, $\gamma >0$, then there exist
two measurable functions $p(u)$ and $b(u)$ of $u\in(0,1)$ such that $\sup
(\left| p(u)\right| ,\left| b(u)\right| )\rightarrow 0$ as $u\rightarrow 0$
and a positive constant c so that 
\begin{equation}
G^{-1}(1-u)=\log c+\log (1+p(u))-\gamma \log u+(\int_{u}^{1}b(t)t^{-1}dt),%
\text{ }0<u<1,  \label{rep1}
\end{equation}
where $G^{-1}(u)=\inf\{x,G(x)\geq u\},$ $0< u\leq 1$ is the generalized
inverse of $G$ with $G^{-1}(0)=G^{-1}(0+)$.\newline
\newline
\noindent (b) If $F\in D(G_{\gamma})$, $\gamma <0$, $\gamma >0$, then $%
y_{0}(G)=\sup \{x,$ $G(x)<1\}<+\infty $ and there exist two measurable
functions $p(u)$ and $b(u)$ for $u\in(0,1)$ and a positive constant c as
defined in (\ref{rep1}) such that

\begin{equation}
y_{0}-G^{-1}(1-u)=c(1+p(u))u^{-\gamma }\exp (\int_{u}^{1}b(t)t^{-1}dt),\text{ 
}0<u<1.  \label{rep2}
\end{equation}

\item  Representation of de Haan (Theorem 2.4.1 in \cite{dehaan}) :

\noindent If $G\in D(\Lambda )$, then there exist two measurable functions $%
p(u)$ and $b(u)$ of $u\in(0,1)$ and a positive constant c as defined in (\ref
{rep1}) such that for 
\begin{equation}
s(u)=c(1+p(u))\exp (\int_{u}^{1}b(t)t^{-1}dt),\text{ }0<u<1,  \label{rep3b}
\end{equation}
we have for some constant $d\in \mathcal{R}$,

\begin{equation}
G^{-1}(1-u)=d-s(u)+\int_{u}^{1}s(t)t^{-1}dt,\text{ }0<u<1.  \label{rep3a}
\end{equation}
\end{enumerate}
\end{theorem}

\bigskip \noindent It is important to remark at once that any $df$ $F$ in the
domain of attraction is associated with a couple of functions $(p,b)$ linked to $G(x)=F(e^{x})$ used
in each appropriate representation. Our conditions will rely on them. In
order to state them, we set first for $0<\tau \leq 1/2$ 
\begin{equation}
a_{n}(\tau )=k^{-\tau }\sum_{j=1}^{k}j^{\tau -1}  \label{defan}
\end{equation}
and 
\begin{equation}
\sigma _{n}^{2}(\tau )=k^{-2\tau }\sum_{j=1}^{k}j^{2(\tau -1)}.
\label{defsigma}
\end{equation}
From there, our conditions are the following, where $\lambda $ is an
arbitrary real number greater than one. The first is 
\begin{equation}
g_{1,n}(p,\lambda)(\sigma _{n}(\tau )k^{\tau })^{-1}\sum_{j=1}^{k}j^{\tau
}\rightarrow 0 \:\: as \: n\rightarrow +\infty  \tag{C1}
\end{equation}
for 
\begin{equation*}
0\leq g_{1,n}(p,\lambda )=\sup_{0\leq u\leq \lambda k/n}\left| p(u)\right| .
\end{equation*}
The second is 
\begin{equation}
g_{2,n}(p,\lambda )(\sigma _{n}(\tau )k^{\tau })^{-1}\sum_{j=1}^{k}j^{\tau
-1}\rightarrow 0\:\: as \: n\rightarrow +\infty  \tag{C2}
\end{equation}
for 
\begin{equation*}
g_{2,n}(b,\lambda )=\sup_{0\leq u\leq \lambda k/n}\left| b(u)\right|.
\end{equation*}
The third is 
\begin{equation}
d_{n}(p,b,\lambda )(\sigma _{n}(\tau )k^{\tau })^{-1}\sum_{j=1}^{k}j^{\tau
-1}\rightarrow 0 \: \: as \: n\rightarrow +\infty  \tag{C3}
\end{equation}
for $d_{n}(p,b,\lambda )=\max (g_{1,n}(p,\lambda ),g_{2,n}(b,\lambda )\log
k)\rightarrow 0$ $as$ $n\rightarrow \infty .$

\bigskip

\noindent We finally point out that all the limits in the sequel are meant
as $n \rightarrow +\infty$ unless the contrary is specified. We are now able
to state our results in the following section.

\section{The results}

\label{results}

\label{laws}

\begin{theorem}
\label{theo2} \bigskip Let $F\in D(\Lambda)=D(G_{0}).$ Then, if (C3) holds, 
\begin{equation*}
(s(k/n)\sigma _{n}(\tau ))^{-1}(T_{n}(\tau )-a_{n}(\tau )s(k/n))\rightarrow 
\mathcal{N}(0,1)
\end{equation*}
for $\tau =1/2$ and 
\begin{equation*}
(s(k/n)\sigma _{n}(\tau ))^{-1}(T_{n}(\tau )-a_{n}(\tau )s(k/n))\rightarrow
\mathcal{L}(\tau ),
\end{equation*}
for $0<\tau <1/2$, and where 
\begin{equation*}
\mathcal{L}(\tau )=\zeta (2(1-\tau ))^{-1/2}\sum_{j=1}^{\infty }j^{-(1-\tau )}(E_{j}-1)
\end{equation*}
is a centered and normed random variable and 
\begin{equation*}
\zeta (s)=\sum_{j=1}^{\infty }j^{-s},\: s>1,
\end{equation*}
is the Riemann function.\\

\noindent Let $F\in D(G_{\gamma })$, $0<\gamma<+\infty$. Then if $(C1)$ and 
$(C2)$ hold, 
\begin{equation*}
(a_{n}(\tau )/\sigma _{n}(\tau ))(T_{n}(\tau )/a_{n}(\tau )-\gamma
)\rightarrow \mathcal{N}(0,\gamma ^{2})
\end{equation*}
for $\tau =1/2$ and 
\begin{equation*}
(a_{n}(\tau )/\sigma _{n}(\tau ))(T_{n}(\tau )/a_{n}(\tau )-\gamma
)\rightarrow \gamma \mathcal{L}(\tau ).
\end{equation*}
\end{theorem}

\begin{remark}
\bigskip We should remark that the methods used here apply for $0<\tau <1.$
Moreover, these results do not extend to parameters $\tau \in ]1/2,1[$ since
for such ones, $k^{\tau }\sigma _{n}(\tau )\rightarrow 0,$ while our results
are based on the hypothesis $k^{\tau }\sigma _{n}(\tau )\rightarrow C\in
]0,+\infty ].$
\end{remark}

\begin{proof}
Let us use the representations (\ref{repre}). We have, for any $n\geq 1,$%
\begin{equation*}
\left\{ \log X_{n-j+1,n}=Y_{n-j+1,n},1\leq j\leq n\right\} =_{d}\left\{
G^{-1}\{1-U_{j,n}),1\leq j\leq n\right\} .
\end{equation*}
Our proofs will rely on Lemma \ref{lemmatool} below which establishes that 
\begin{equation*}
\left\{ k^{-\tau }\sum_{j=1}^{b}j^{\tau }\log (\frac{U_{j+1,n}}{U_{j,n}}%
)-a_{n}(\tau )\right\} /\sigma _{n}(\tau )=_{d}(k^{-\tau }\sigma _{n}(\tau
))^{-1}\sum_{j=1}^{k}j^{\tau -1}(E_{j}-1)
\end{equation*}
converges in distribution to $\mathcal{N}(0,1)$ random variable for $\tau =1/2$ and to
the finite random variable 
\begin{equation*}
\mathcal{L}(\tau )=A(2(1-\tau ))^{-1/2}\sum_{j=1}^{\infty }j^{-(1-\tau )}(E_{j}-1)
\end{equation*}
for $0<\tau <1/2,$ where the independent standard exponential random
variables $E_{1},E_{2},...$ are defined in (\ref{malm}). We begin by letting 
$F\in D(\Lambda)=D(G_{1/\gamma})$ for $\gamma=+\infty$. By (\ref{rep2}), we get 
\begin{equation*}
T_{n}(\tau )=k^{-\tau }\sum_{j=1}^{k}j^{\tau
}(s(U_{j,n})-s(U_{j+1,n}))+k^{-\tau }\sum_{j=1}^{k}j^{\tau
}\int_{U_{j,n}}^{U_{j+1,n}}s(t)/t\text{ }dt
\end{equation*}

\begin{equation}
\equiv S_{n}(1)+S_{n}(2).  \label{pr4a}
\end{equation}
By using (\ref{rep3b}), we have for $U_{1,n}\leq v,u\leq U_{k,n},$%
\begin{equation*}
s(u)/s(v)=(1+p(u))/(1+p(v))\exp (-\int_{u}^{v}t^{-1}b(t)dt).
\end{equation*}
By putting 
\begin{equation}
g_{1,n,0}(p)=\sup \{\left| p(u)\right| ,0\leq u\leq U_{k+1,n}\}\text{ and }%
g_{2,n,0}(p)=\sup \{\left| b(u)\right| ,0\leq u\leq U_{k+1,n}\},
\label{preuve3}
\end{equation}
and by remarquing that $\log (U_{k+1,n}/U_{1,n})=O_{p}(\log k)$, we get 
\begin{equation*}
s(u)/s(v)=(1+O(g_{1,n,0}))\exp (-O_{p}(g_{2,n,0}\log k)),
\end{equation*}
whenever 
\begin{equation}
g_{2,n,0}\log k\rightarrow _{P}0,  \label{cond00}
\end{equation}
and next
\begin{equation}
\sup_{U_{1,n}\leq u,v\leq U_{k,n}}\left| s(u)/s(v)-1\right| =O_{p}(\max
(g_{1,n,0},g_{2,n,0}\log k))  \label{preuve2}
\end{equation}
and finally, 
\begin{equation}
\sup_{U_{1,n}\leq u,v\leq U_{k,n}}\left| \frac{s(u)-s(v)}{s(k/n)}\right|
=O_{p}(\max (g_{1,n,0},g_{2,n,0}\log k).  \label{preuve0}
\end{equation}
We have then to prove (\ref{cond00}). But since $nk^{-1}U_{k+1,n}\rightarrow
1$, for a fixed value $\lambda >1$ such that (C3) holds, we may find for any $%
\varepsilon >0$\ $,$ an integer $N_{0}$ such that for any $n\geq N_{0},$%
\begin{equation*}
\mathbb{P}(g_{2,n,0}\log k\geq g_{2,n}\log k)\equiv \mathbb{P}(A_{n})<\varepsilon .
\end{equation*}
Now fix $\delta >0.$ By using the sets decomposition \ 
\begin{equation*}
(\left| g_{2,n,0}\log k\right| >\delta )=(\left| g_{2,n,0}\log k\right|
>\delta )\cap A_{n}+(\left| g_{2,n,0}\log k\right| >\delta )\cap A_{n}^{c},
\end{equation*}
we have, for $n\geq N_{0}$ 
\begin{equation*}
\mathbb{P}(\left| g_{2,n,0}\log k\right| >\delta )\leq \mathbb{P}(\left| g_{2,n}\log k\right|
>\delta )+\varepsilon .
\end{equation*}
By letting first $n\rightarrow \infty $ and next $\varepsilon \downarrow 0$, and by taking $(C3)$ into account,
we see that (\ref{cond00}) holds and so does (\ref{preuve0}). 
By putting together (\ref{pr4a}) and (\ref{preuve0}), and by using the
preceding techniques, we get for a picked value $\lambda >1$ like in (C$3)$,
for any $\varepsilon >0,$ a value N$_{0}$ such that for any $n\geq N_{0}$
and for any $\delta >0,$ 
\begin{equation*}
\mathbb{P}(\left| \frac{S_{n}(1)}{\sigma _{n}(\tau )s(k/n)}\right| >\delta )\leq
\mathbb{P}(d_{n}(p,b,\lambda )\text{ }(k^{\tau }\sigma _{n}(\tau
))^{-1}\sum_{j=1}^{k}j^{\tau }>\delta )+\varepsilon ,
\end{equation*}
where $d_{n}(p,b,\lambda )=\max (g_{1,n},g_{2,n}\log k).$ Again, by letting
first $n\uparrow \infty $ and next $\varepsilon \downarrow 0$, and by using $(C3)$, we get that  $%
S_{n}(1)(\sigma _{n}(\tau )s(k/n))^{-1}\rightarrow _{P}0.$ We now treat $%
S_{n}(2)$ $\ $and remark that 
\begin{equation*}
\frac{S_{n}(2)}{\sigma _{n}(\tau )s(k/n)}=(k^{\tau }\sigma _{n}(\tau
))^{-1}\sum_{j=1}^{k}j^{\tau }\int_{U_{j,n}}^{U_{j+1,n}} t^{-1} \left\{
s(t)/s(k/n)\right\} \text{ }dt
\end{equation*}
\begin{equation*}
=(k^{\tau }\sigma _{n}(\tau ))^{-1}\sum_{j=1}^{k}j^{\tau
}\int_{U_{j,n}}^{U_{j+1,n}}t^{-1}\text{ }dt
\end{equation*}
\begin{equation*}
+(k^{\tau }\sigma _{n}(\tau ))^{-1}\sum_{j=1}^{k}j^{\tau
}\int_{U_{j,n}}^{U_{j+1,n}}\left\{ s(t)/s(k/n)-1\right\} /t\text{ }%
dt=S_{n}(2,1)+S_{n}(2,2).
\end{equation*}
We have, by (\ref{preuve2}) and the Malmquist representation (\ref{malm}), 
\begin{equation*}
\left| S_{n}(2,2)\right| \leq O_{p}(1)d_{n}(p,b,\lambda )(k^{\tau }\sigma
_{n}(\tau ))^{-1}\sum_{j=1}^{k}j^{\tau -1}E_{j}\leq O_{p}(1)\text{ }\times
\end{equation*}
\begin{equation*}
\left\{ d_{n,0}(p,b,\lambda )(k^{\tau }\sigma _{n}(\tau
))^{-1}\sum_{j=1}^{k}j^{\tau -1}(E_{j}-1)+d_{n,0}(p,b,\lambda )(k^{\tau
}\sigma _{n}(\tau ))^{-1}\sum_{j=1}^{k}j^{\tau -1}\right\} .
\end{equation*}
where $d_{n,0}(p,b,\lambda)=\max (g_{1,n,0},g_{2,n,0}\log k).$ The first term tends to zero since 
$d_{n,0}(p,b,\lambda )\rightarrow _{P}0$ by (\ref{cond00})$,$ and $(k^{\tau
}\sigma _{n}(\tau ))^{-1}\sum_{j=1}^{k}j^{\tau -1}(E_{j}-1)$ converges in
distribution to a finite random variable by Lemma \ref{lemmatool} in Section 
\ref{sectiontool}. We handle the second term also by the previous
techniques. We choose a value $\lambda >1$ like in $(C3).$ For any $%
\varepsilon >0,$ we get value N$_{0}$ such that for any $n\geq N_{0}$ and
for any $\delta >0,$%
\begin{equation*}
\mathbb{P}(\left| d_{n,0}(p,b,\lambda )(k^{\tau }\sigma _{n}(\tau
))^{-1}\sum_{j=1}^{k}j^{\tau -1}\right| >\delta )\leq \varepsilon + \mathbb{P}(\left|
d_{n}(p,b,\lambda )(k^{\tau }\sigma _{n}(\tau ))^{-1}\sum_{j=1}^{k}j^{\tau
-1}\right| >\delta ).
\end{equation*}
By letting first $n\uparrow \infty $ and next $\varepsilon \downarrow 0$, we
see by $(C3)$ that this term goes to $0$ in probability. Finally, by the Malmquist
representation (\ref{malm}), one arrives at 
\begin{equation*}
S_{n}(2,1)=(k^{\tau }\sigma _{n}(\tau ))^{-1}\sum_{j=1}^{k}j^{\tau -1}E_{j}+o_{P}(1).
\end{equation*}

\noindent And this leads to 
\begin{equation*}
S_{n}(2,1)-a_{n}(\alpha )/\sigma _{n}(\tau )=(k^{\tau }\sigma _{n}(\tau
))^{-1}\sum_{j=1}^{k}j^{\tau -1}(E_{j}-1)\equiv V_{n}(\tau )
\end{equation*}
and finally, by summing up what precedes, 
\begin{equation*}
(\sigma _{n}(\tau )s(k/n))^{-1}(T_{n}(\tau )-a_{n}(\tau ))=V_{n}(\tau
)+o_{P}(1).
\end{equation*}
This concludes the proof since, by Lemma \ref{lemmatool} in Section \ref
{sectiontool}, $V_{n}(\tau )$ converges in probability to a $\mathcal{N}%
(0,1)$ random variable for $\tau =1/2$ and to $\mathcal{L}(\tau )$ for $\tau <1/2$.\\

\noindent Now let $F\in D(G_{\gamma}), \text{ } 0<\gamma<+\infty$. We will have similar proofs but
we do not explicit the convergences in distribution as we did in the previous
case. We have by (\ref{rep1}) and the usual representations, 
\begin{equation*}
T_{n}(\tau )=k^{-\tau }\sum_{j=1}^{k}j^{\tau }\left\{ \log
(1+p(U_{j+1,n}))-\log (1+p(U_{j,n}))\right\}
\end{equation*}
\begin{equation*}
+\gamma k^{-\tau }\sum_{j=1}^{k}j^{\tau }\log (U_{j+1,n}/U_{j,n})+k^{-\tau
}\sum_{j=1}^{k}j^{-\tau }\int_{U_{j,n}}^{U_{j+1,n}}b(t)/t\text{ }dt
\end{equation*}

\begin{equation*}
\equiv S_{n}(1)+S_{n}(2)+S_{n}(3).
\end{equation*}
We have, for large values of $k$, 
\begin{equation*}
\left| S_{n}(1)/\sigma _{n}(\tau )\right| \leq 2g_{1,n,0}(f)(\sigma
_{n}(\tau )k^{\tau })^{-1}\sum_{j=1}^{k}j^{\tau },
\end{equation*}
where $g_{1,n,0\text{ }}$is defined in (\ref{preuve3}), which tends to zero
in probability by $(C1)$ and (\ref{preuve2}). Next 
\begin{equation*}
\left| S_{n}(3)/\sigma _{n}(\tau )\right| \leq g_{2,n,0}(b)(\sigma _{n}(\tau
)k^{\tau })^{-1}\sum_{j=1}^{k}j^{\tau }\log (U_{j+1,n}/U_{j,n})
\end{equation*}
\begin{equation*}
=g_{2,n,0}(b)(\sigma _{n}(\tau )k^{\tau })^{-1}\sum_{j=1}^{k}j^{\tau
-1}(E_{j}-1)+g_{2,n,0}(b)(\sigma _{n}(\tau )k^{\tau
})^{-1}\sum_{j=1}^{k}j^{\tau -1},
\end{equation*}
where $g_{2,n,0\text{ }}$defined in (\ref{preuve3}). Then $S_{n}(3)/\sigma
_{n}(\tau )\rightarrow 0$ by $(C2)$, Lemma \ref{lemmatool} and the
methods described above. Finally, always by Lemma \ref{lemmatool}, we get 
\begin{equation*}
\left\{ (S_{n}(3)-\gamma a_{n}(\tau )\right\} /\sigma _{n}(\tau )=\gamma
(\sigma _{n}(\tau )k^{\tau })^{-1}\sum_{j=1}^{k}j^{\tau -1}(E_{j}-1)
\end{equation*}
and next 
\begin{equation*}
\left\{ T_{n}(\tau )-\gamma a_{n}(\tau )\right\} /\sigma _{n}(\tau )=\gamma
(\sigma _{n}(\tau )k^{\tau })^{-1}\sum_{j=1}^{k}j^{\tau -1}(E_{j}-1)+o_{P}(1)
\end{equation*}
\begin{equation*}
=\gamma V_{n}(\tau )+o_{P}(1)
\end{equation*}
and this converges in distribution to $\mathcal{N}(0,\gamma ^{2})$ random
variable for $\tau =1/2$ and to $\gamma \mathcal{L}(\tau )$ for $\tau <1/2.$ The proof
is completed with the remark that $a_{n}(\tau )\rightarrow \tau ^{-1}.$
\end{proof}

\section{Remarks, Applications and Simulations}

\label{remarks}

\subsection{An asymptotically Gaussian estimator}

Now let us illustrate some statistical applications of our results. Let $%
F\in D(G_{\gamma}$). Recall that for $\tau >1/2,$ $\tau T_{n}(\tau )$
has been proved to be asymptotic normal estimators of $\gamma .$ Our present
results extend those obtained to  $\tau =1/2$ and show that this latter value
is a critical one since $\tau T_{n}(\tau )$ is still an estimator of $\gamma $
for $\tau <1/2\ $but has a non-Gaussian asymptotic behavior. Now let us
particularize the results for each case and illustrate the conditions $(C1)$
and $(C2)$. For $\tau =1/2,$\ we have under the $(C1)$ and $(C2)$, 
\begin{equation*}
2k^{1/2}(\log k)^{-1/2}(T_{n}(1/2)/a_{n}(1/2)-\gamma )\rightarrow
N(0,\gamma ^{2}), \label{an01}
\end{equation*}
Thus $T_{n}(1/2)/a_{n}(1/2)$ is a new estimator of the extreme value index
which is asymptotically normal. Since $a_{n}(1/2)\rightarrow 1/2,$ we
rediscover the general result $\tau T_{n}(\tau )\rightarrow \gamma $. By $%
(S4)$ below, $a_{n}(1/2)=(1/2)(1+O(k^{-1/2}))$ so that we also have 
\begin{equation*}
2k^{1/2}(\log k)^{-1/2}(2T_{n}(1/2)-\gamma )\rightarrow N(0,\gamma ^{2}).
\end{equation*}
From there we may form another family of estimators $%
T_{n}(1/2,a)=(2a+(1-a)/a_{n}(1/2))T_{n}(1/2),$ for $0<a<1,$ which are
asymptotically normal since 
\begin{equation}
2k^{1/2}(\log k)^{-1/2}(T_{n}(1/2,a)-\gamma )\rightarrow N(0,\gamma ^{2}).
\end{equation}

\noindent This parameter $a$, that may be randomly picked from (0,1), may be usefull in
applications. The condition $(C1)$
turns to 
\begin{equation*}
g_{1,n}(p,\lambda )(\log k)^{-1/2}\sum_{j=1}^{k}j^{1/2}\sim
g_{1,n}(p,\lambda )k^{-1/2}(\log k)^{-1/2}\rightarrow 0.
\end{equation*}
But in many real situations, we may suppose that the $df$ admits a
derivative so that we may take $p(u)=0$, which renders $(C1)$ useless.
Indeed, suppose that $F(x)$ admits a derivative $F^{\prime}(x)$ for large
values of x. A condition for $F$ to belong to $D(G _{1/\gamma }), \: 0<\gamma<+\infty$, is 
\begin{equation}
\lim_{x\rightarrow \infty }\frac{1-F(x)}{xF^{\prime }x)}=\frac{1}{\gamma }.
\label{CE01}
\end{equation}
(see \cite{dehaan1}, p. 17, Theorem 1.1.11). This implies that 
\begin{equation*}
b(u)=u(G^{-1}(1-u))^{\prime }+\gamma =uG^{\prime }(G^{-1}(1-u))^{-1}+\gamma
\rightarrow 0\text{ as }n\rightarrow \infty
\end{equation*}
and next, for some constant $c>0$ and $u_{0}\in (0,1),$%
\begin{equation*}
G^{-1}(1-u)=c-\gamma \log u+\int_{u}^{u_{0}}b(t)t^{-1}dt.
\end{equation*}
For such $df^{\prime}s$, $(C1)$ is effectively useless. The
condition $(C2)$ is more consistent. However, it is somewhat related to the
extreme second order condition (see \cite{dehaan1}, p.43) and some popular models
used in the literature. For example, in the Hall one, that is $%
1-F(x)=C_{1}x^{-1/\gamma }(1+C_{2}x^{-\beta }+o(x^{-\beta })),$ where $\beta$, $C_{1}$
and $C_{2}$ are positive constants, we have for constants $C_{i},\: i=1,2,...$%
, $F^{-1}(1-u)=C_{3}u^{-\gamma }(1+C_{4}u^{\beta \gamma }+o(u^{\beta \gamma
}))$ $\ $and then 
\begin{equation*}
G^{-1}(1-u)=\log C_{3}-\gamma \log u+O(u^{\beta \gamma }),
\end{equation*}
which may be written, with $b_{1}(t)=u^{\beta \gamma }$ and $c_{1}(u)=O(1)$
as $u\rightarrow 0,$%
\begin{equation*}
G^{-1}(1-u)= C_{4}-\gamma \log u+c_{1}(u)\int_{u}^{1}(b_{1}(t)/t) \: dt.
\end{equation*}
Suppose now that $c_{1}(u)=D$ or that $c_{1}(u)\int_{u}^{1}(b_{1}(t)/t) \text{ } dt$ is
obtained from $\int_{u}^{1}c(t)(b_{1}(t)/t) \text{ } dt$ by the Integral Mean Theorem
with a bounded function $c(\cdot)$. Then condition $(C2)$ turns to 
\begin{equation*}
k^{\beta \gamma +1/2}(\log k)^{-1/2}/n^{\beta \gamma }\rightarrow 0
\end{equation*}
and $(C3)$\ to 
\begin{equation*}
k^{\beta \gamma +1/2}(\log k)^{1/2}/n^{\beta \gamma }\rightarrow 0.
\end{equation*}
We conclude that for $\tau =1/2$, $T_{n}(1/2)/a_{n}(1/2)$ is an asymptotic
normal estimator of $\gamma $ under reasonable assumptions based on the regularity
on the $df$ $F$. Simulation studies suggest its good performances as well.

\subsection{A class of non asymptotically Gaussian estimators} \label{subsec3}

For $0<\tau <1/2,$ we also have under $(C1)$ and $(C2)$, 
\begin{equation}
\frac{k^{\tau \gamma }}{\tau \zeta (2(1-\tau ))}(T_{n}(\tau )/a_{n}(\tau )-\gamma
)\rightarrow  \mathcal{L}(\tau ).  \label{ng01}
\end{equation}
Here again, we get a family of estimators of $\gamma .$ However, the
limiting law is not Gaussian. This seems to be new. Remark here that we
cannot simplify (\ref{ng01}) to 
\begin{equation*}
\frac{k^{\tau \gamma}}{\tau \zeta (2(1-\tau ))}(\tau T_{n}(\tau )-\gamma
)\rightarrow \mathcal{L}(\tau )
\end{equation*}
like for $\tau =1/2$ because of $\sigma _{n}(\tau )$ being finite and
because of $(S3)$ below.

\noindent This results make a connexion of Extreme Value Theory and Number Theory
since they deeply depend of the Riemann function closely related to the
prime numbers since, for $s>1,$ 
\begin{equation*}
\zeta (s)=\prod_{p=2}^{\infty }(1-p^{-s})^{-1}=\frac{s}{s-1}+\frac{1}{s}%
\int_{1}^{\infty }\frac{x-[x]}{x^{s+1}}dx,
\end{equation*}
where the product extends to the prime numbers $p\geq 2$ and $[x]$\ denotes
the integer part of \ $x$. The limiting law $\mathcal{L}(\tau )$ is
characterized by its characteristic function 
\begin{equation*}
\psi _{\infty }(t)=\exp (\sum_{n=2}^{\infty }\frac{(it)^{n}}{n}\zeta
(n(1-\tau ))\zeta (2(1-\tau ))^{-n/2}),
\end{equation*}
calculated and justified in (\ref{fc01}). This random variable has all its
moments finite and is related to the Riemann function. Indeed put $\psi
_{\infty }(t)=\exp (\phi (t))$ and let $\phi ^{(r)}$ denote the $r^{th}$
derivative function of $\phi .$ We easily see that 
\begin{equation*}
\psi _{\infty }^{(1)}(t)=\phi ^{(1)}(t)\exp (\phi (t)),\text{ with }\phi
^{(1)}(t)=\sum_{n=2}^{\infty }i(it)^{n-1}\zeta (n(1-\tau ))\zeta (2(1-\tau
))^{-n/2}),
\end{equation*}
and 
\begin{equation*}
\psi _{\infty }^{(2)}(t)=(\phi ^{(1)}(t))^{2}+\phi ^{(2)}(t))\exp (\phi (t)),
\end{equation*}
with 
\begin{equation*}
\phi ^{(2)}(t)=\sum_{n=2}^{\infty }-(n-1)(it)^{n-2}\zeta (n(1-\tau ))\zeta
(2(1-\tau ))^{-n/2}).
\end{equation*}
So 
\begin{equation*}
\phi ^{(1)}(0)=0\text{ and }\phi ^{(2)}(0)=-\left[ (n-1)(it)^{n-2}\zeta
(n(1-\tau ))\zeta (2(1-\tau ))^{-n/2})\right] _{n=2}=-1
\end{equation*}
and thus 
\begin{equation*}
\mathbb{E}X=0, \\ \\ \mathbb{E}X^{2}=1.
\end{equation*}
One may repeat the same methods to get 
\begin{equation*}
\mathbb{E}X^{3}=2\zeta (3(1-\tau ))/\zeta (2(1-\tau ))^{3/2},\mathbb{E}%
X^{4}=3+6\zeta (4(1-\tau ))/\zeta (2(1-\tau ))^{2}
\end{equation*}
and higher moments. For such values of $\tau ,$ the condition $(C2)$ and $%
(C3)$ become 
\begin{equation*}
g_{2,n}(p,\lambda ) \rightarrow 0
\end{equation*}
and 
\begin{equation*}
\max (g_{1,n}(p,\lambda ),g_{2,n}(b,\lambda )) \rightarrow 0.
\end{equation*}

\noindent \textbf{We get here the interesting thing that $(C2)$ always holds for any $df$ $F \in D(G_{\gamma}), 0 < \gamma <+\infty$. And this is a quite general and rare result in EVT. For $\gamma=0$, the condition $(C3)$ is also very weak since $g_{1,n}$ and $g_{2,n}$ are only requested to tend to zero faster that $\log k$. We get then very general asymptotic laws that sound as a compensation of the lack of normality!}\\

\bigskip

\noindent Since these limits seem new in the extreme value index estimation, it may be reasonable that we give
some comments on how to use them. Although we cannot give a simple expression of the $df$ of $\mathcal{L}(\tau)$, 
$F_{\mathcal{L}(\tau)}(u)=\mathbb{P}(\mathcal{L}(\tau) \leq u)$, we may use computer-based methods to compute approximations of its values.
We may fix a great value of $k$ and then consider a large sample, of size $B$, of the random variable 
\begin{equation}
\mathcal{L}(\tau,k)=\zeta(2(1-\tau),m)^{-1/2} \sum_{j=1}^{k} j^{\tau-1} (E_{j}-1) \label{laltau}
\end{equation}
where $\zeta(s,m) =\sum_{j=1}^{m} j^{-s}$, with a sufficiently large value of m. We surely get the empirical distribution function ($edf$) of this sample as a fair uniform approximation of $F_{\mathcal{L}(\tau)}$. Fix $m=10000$, $B=10000$, we may observe in the right graph in Figure \ref{fig:FIG. 1}, that the $df$'s of $F_{L(\tau)}$ for $\tau \in \{0.1, 0.2, 0.25, 0.3, 0.25, 0.3, 0.35, 0.40, 0.48, 0.48 \}$ are very close one another. When we consider a $\tau$ uniformly drawn from $]0,1/2[$, we obtain a significantly different $df$ as illustrated in the left graph in Figure 3, where we simply add this last $df$ to those already in the first graph.\\

\noindent From there, we are able to compute the probabilities and the quantiles of such limiting laws, as illustrated in the Table \ref{tab:tab1} in the Appendix. Statistical tests may then be based on these results.

\subsection{Simulations}

We wish to show the performance of these class of statistics as estimators
of $\gamma$ and compare them to some other ones available in the
literature, such as Hill's one $T_{n}(1)$ (see \cite{hill}), the Pickands
estimator (see \cite{pickands}) 

\begin{equation*}
P_{n}(k)=(\log 2)^{-1} \log \left\{ (X_{n-k,n}-X_{n-2k,n})/(X_{n-2k,n}-X_{n-4k,n})\right\} ,
\end{equation*}

\noindent and the Lo estimator (see \cite{lo}) $L_{n}(k)$ with 

\begin{equation*}
{\small L_{n}^{2}(k)=k^{-1}\sum_{j=1}^{k}\sum_{i=j}^{k}jg(i,j)(\log
X_{n-j+1,n}-\log X_{n-j,n})(\log X_{n-i+1,n}-\log X_{n-i,n}),}
\end{equation*}
where $g(i,j)=1/2$ for $i=j$ and $g(i,j)=1$ otherwise. We may see that the
performance of these estimators largely depend on the nature of the tail $%
1-F.$ In the simple case where $1-F(x)=x^{-1/\gamma }$, all these estimators
perform pretty well for $n=300,$ $k\sim n^{0.75}$ for $\gamma=0.5$  and $T_{n}(1/2)/a_{n}(1/2)$ has particularly good results. They also present good
performance for the model $1-F(x)=x^{-1/\gamma }(1+O(x^{-\eta }))$ with $\eta =100$. The results
are less sound for small $\eta$ ($\eta=10$ for example) for sample sizes around $300$. We check
the performances in Table \ref{tab:tab2} in the Appendix, where we report the estimated values for $\gamma$ for the sample case and the MSE, and in Figure \ref{fig:fig2} all in the Appendix. The asymptotic normality  of (\ref{an01}) is illustrated in the graphs of Figure \ref{fig:fig3} as good. 

\noindent The simulation results for the asymptotic normality of (\ref{an01}) are also acceptable since the empirical p-value is high, of order 90\%.\\
\bigskip

\noindent For values $\tau \in ]0,1/2[$, we illustrate the results for $\tau=0.3$. In the graph in right in Figure \ref{fig:fig4}, we realize that $T_{n}(0.3)/a_{n}(0.3)$ is a fair estimate of $\gamma$ in the simple model described above.

\noindent As said previously Figure 3 successfully compares the law of $\frac{k^{\tau \gamma }}{\tau \zeta (2(1-\tau ))}(T_{n}(\tau )/a_{n}(\tau )-\gamma)$ in (\ref{ng01}) and that of (\ref{laltau}).

\section{Lemmas and Tools}

\label{sectiontool}

We begin by this simple lemma where we suppose that we are given a sequence
of independent and uniformly distributed random variables $U_{1},U_{2},...$
as in (\ref{repre}).

\begin{lemma}
\label{lemmatool} \label{lemma1} Let 
\begin{equation*}
V_{n}(\tau )=k^{-\tau }\sum_{j=1}^{b}j^{\tau }\log (\frac{U_{j+1,n}}{U_{j,n}}%
).
\end{equation*}
Then for $\tau =1/2,$%
\begin{equation*}
\sigma _{n}^{-1}(\tau )(V_{n}(\tau )-a_{n}(\tau ))\leadsto \mathcal{N}(0,1)
\end{equation*}
and for $0<\tau <1/2,$%
\begin{equation*}
\sigma _{n}^{-1}(\tau )(V_{n}(\tau )-a_{n}(\tau ))\leadsto \mathcal{L}(\tau ),
\end{equation*}
where $a_{n}(\tau )$ and $\sigma _{n}(\tau )$ are defined in (\ref{defan})
and \ref{defsigma}), and 
\begin{equation*}
\mathcal{L}(\tau )=A(2(1-\tau ))^{-1/2}\sum_{j=1}^{\infty }j^{-(1-\tau )}(E_{j}-1),
\end{equation*}
is a centred and reduced random variable having all moments finite.
\end{lemma}

\begin{proof}
By using the Malmquist representation (\ref{malm}), we have 
\begin{equation*}
V_{n}(k)=k^{-\tau }\sum j^{\tau -1}E_{j}.
\end{equation*}
It follows that $\mathbb{E}(V_{n}(k))=a_{n}(k)$ and $\mathbb{V}%
ar(V_{n}(k))=\sigma _{n}^{2}(k).$ Put

\begin{equation*}
V_{n}^{\ast }(\tau )=\sigma _{n}(k)^{-1}(V_{n}(k)-a_{n}(k)).
\end{equation*}
Then

\begin{equation*}
V_{n}^{\ast }(\tau )=(k^{\tau }\sigma _{n}(\tau ))^{-1}\sum_{j=1}^{k}j^{\tau
-1}(E_{j}-1).
\end{equation*}
For $0<\tau <1/2,$ we have by $(S1)$ below, $k^{\tau }\sigma _{n}(\tau
)\rightarrow A(2(1-\tau ))^{-1/2}$ and 
\begin{equation*}
V_{n}^{\ast }(\tau )\rightarrow \zeta (2(1-\tau ))^{-1/2}\sum_{j=1}^{\infty
}j^{\tau -1}(E_{j}-1)=L(\tau ).
\end{equation*}
Now, we have to prove that $\mathcal{L}(\tau )$ is a well-defined random variable with
all finite moments. The characteristic function of $V_{n}^{\ast
}(\tau )$ is 
\begin{equation*}
\psi _{V_{n}^{\ast }(\tau )}(t)=\exp (-\zeta (2(1-\tau
))^{-1/2}\sum_{j=1}^{k}j^{\tau -1}(it))\prod_{j=1}^{k}(1-itj^{-(1-\tau
)}\zeta (2(1-\tau ))^{-1/2})^{-1}.
\end{equation*}
By using the development of $\log (1-\cdot )$ and by the Lebesgues Theorem,
one readily proves that 
\begin{equation*}
\psi _{V_{n}^{\ast }(\tau )}(t)=\exp (\sum_{j=1}^{k}\sum_{n=2}^{\infty }%
\frac{(it)^{n}}{n}j^{-n(1-\tau )}\zeta (2(1-\tau ))^{-n/2})
\end{equation*}
\begin{equation}
\rightarrow \psi _{\infty }(t)=\exp (\sum_{n=2}^{\infty }\frac{(it)^{n}}{n}%
\zeta (n(1-\tau ))\zeta (2(1-\tau ))^{-n/2}).  \label{fc01}
\end{equation}
Recall as in \cite{valiron}, p.506, that for $s>1,$%
\begin{equation*}
\zeta (s)=\frac{s}{s-1}+\frac{1}{s}\int_{1}^{\infty }\frac{x-[x]}{x^{s+1}}dx,
\end{equation*}
where $[x]$ denotes the integer part of $x.$ This leads to 
\begin{equation*}
s(s-1)^{-1}\leq \zeta (s)\leq s(1+(s-1)^{-1}).
\end{equation*}
By using this, we see that the absolute value of the general term of the
series in (\ref{fc01}) is dominated, for large values of $n$, as follows 
\begin{equation*}
\left| \frac{(it)^{n}}{n}\zeta (n(1-\tau ))\zeta (2(1-\tau ))^{-n/2}\right|
\leq 2(1-2\tau )(\frac{\left| t\right| }{2(1-\tau )})^{n}.
\end{equation*}
This shows that $\psi _{\infty }(t)$ is well defined and characterizes the 
$df$ of $\mathcal{L}(\tau )$. In Section \ref{remarks}, we anticipated by remarking
that $\mathcal{L}(\tau )$ is centered and has all its moment finite and indicated the
way to get them.

Its remains to handle the case $\tau =1/2$. Let us evaluate the moment
generating function of $V_{n}^{\ast }(1/2):$

\begin{equation*}
\psi _{V_{n}^{\ast }(1/2)}(t)=\prod_{j=1}^{k}\psi _{(E_{j}-1)}(ti^{\tau
-1}(k^{1/2}\sigma _{n}(1/2))^{-1}),
\end{equation*}

\noindent But we have $\sigma _{n}^{2}(1/2)=S(k,1)\sim k^{-1}log k$ and next 
$k^{1/2}\sigma _{n}(1/2)=(\log k)^{1/2}\rightarrow \infty .$ It follows that
for a fixed, for k large enough, 
\begin{equation}
\left| t\text{ }i^{-1/2}\text{ }(k^{1/2}\sigma _{n}(1/2))^{-1}\right| \leq
2t(\log k)^{-1/2}  \label{dl10c}
\end{equation}
uniformly in $i\geq 1$. At this step, we use the expansion of $\psi
_{(E_{j}-1)}$ in the neighbourhood of zero : 
\begin{equation*}
\psi _{(E_{j}-1)}(u)=1+u^{2}/2+u^{4}g(u),
\end{equation*}
where there exists $u_{0}$ such that 
\begin{equation*}
0\leq u\leq u_{0}\Rightarrow \left| g(u)\right| \leq 1.
\end{equation*}
By using the uniform bound in (\ref{dl10c}), for a value of $k$ such that $%
2t(\log k)^{1/2}\leq u_{0},$ we get 
\begin{equation*}
\psi _{(E_{j}-1)}(t\text{ }j^{-1/2}\text{ }(k^{1/2}\sigma _{n}(1/2))^{-1})=1+%
\frac{1}{2}(t\text{ }j^{-1/2}\text{ }(k^{1/2}\sigma _{n}(1/2))^{-1})^{2}
\end{equation*}
\begin{equation*}
+(t\text{ }j^{-1/2}\text{ }(k^{1/2}\sigma _{n}(1/2))^{-1})^{4}g_{0,j,n}(t),
\end{equation*}
where $\left| g_{0,j,n}(t)\right| \leq 1$ for all $1\leq j\leq k.$ By the
uniform boundedness of the error term, we have 
\begin{equation*}
\log \psi _{(E_{j}-1)}(t\text{ }j^{-1/2}\text{ }(k^{1/2}\sigma
_{n}(1/2))^{-1})=
\end{equation*}
\begin{equation*}
\frac{1}{2}(t\text{ }j^{-1/2}\text{ }(k^{1/2}\sigma _{n}(1/2))^{-1})^{2}+(t%
\text{ }j^{-1/2}\text{ }(k^{1/2}\sigma _{n}(1/2))^{-1})^{4}g_{0,j,n}(t)
\end{equation*}
\begin{equation*}
(t\text{ }j^{-1/2}\text{ }(k^{1/2}\sigma _{n}(1/2))^{-1})^{4}g_{1,j,n}(t),
\end{equation*}
where always $\left| g_{1,j,n}(t)\right| \leq 1$ for all $1\leq j\leq k.$
Finally 
\begin{equation*}
\psi _{V_{n}^{\ast }(1/2)}(t)=\exp (\sum_{j=1}^{k}\log \psi
_{(E_{j}-1)}(ti^{\tau -1}(k^{1/2}\sigma _{n}(1/2))^{-1})
\end{equation*}
\begin{equation*}
=\exp (t^{2}/2+t^{4}g_{2}\sum_{j=1}^{k}j^{-2}),
\end{equation*}
where for large values of $k$, 
\begin{equation*}
\left| g_{2}\right| \leq (\log k)^{-2}.
\end{equation*}
Hence 
\begin{equation*}
\psi _{V_{n}^{\ast }(1/2)}(t)\rightarrow \exp (t^{2}/2)
\end{equation*}
and 
\begin{equation*}
V_{n}^{\ast }(1/2)\rightarrow \mathcal{N}(0,1),
\end{equation*}
which achieves the proofs.
\end{proof}

\bigskip

\begin{lemma}
Define for $\tau >0$, $S(k,\tau )=k^{-\tau }\sum_{j=1}^{k}j^{-\tau }.$ Then%
\newline
for $\tau >1$, 
\begin{equation}
S(k,\tau )\sim \zeta (\tau )k^{-\tau },  \tag{S1}
\end{equation}
for $\tau =1/2$, 
\begin{equation}
S(k,\tau )=2(1+O(k^{-1/2})),  \tag{S2}
\end{equation}
for $\tau <1$, 
\begin{equation}
S(k,\tau )=k^{1-2\tau }/(1-\tau )(1+O(k^{-1-\tau }))  \tag{S3}
\end{equation}
and for $\tau =1,$%
\begin{equation}
S(k,\tau )=(\log k)k^{-1}(1+O(1/\log k)).  \tag{S4}
\end{equation}
\end{lemma}

\begin{proof}
Theses formulas are readily obtained by comparing the $\sum_{j=1}^{k}j^{-%
\tau }$ and $\int_{1}^{k}x^{-\tau }dx,$ as in classical analysis.
\end{proof}

\section{Conclusion}

\label{conclusion}

The family $\{T_{n}(\tau ),0<\tau \leq 1/2\}$ is only studied for $F$ belonging to 
$D(G_{\gamma})$, $\gamma \geq 0$. This includes the negative case under the appropriate transform. However the remaining case, that is the Weibull domain, presents a radically different approach including sums of dependent random variables and deserves a seperated study. Furthermore, we still have to develop a Bayesian approach by considering a random choice of the parameter $\tau$ leading to an optimal choice among a large class of
admissible laws.

\newpage
\noindent \LARGE{APPENDIX}.
\Large We give in this appendix all the tables and figures of the paper.

\small
\begin{table}[htbp]
	\centering
		\begin{tabular}{|l|l|l|l|l|l|l|l|l|l|l|}
\hline
                     &   0.1  &  0.15    &   0.2  & 0.25    &  0.3    & 0.35   &  0.4   & 0.45  &   0.48   &  unif\\
\hline
$F(-1.96)$ &0.004   &0.0035    & 0.006  &0.007   & 0.009    &0.0137  & 0.016  &0.0166 & 0.018    & 0.067 \\
\hline
$F(0)$     &0.5572  &0.557     &0.549    &0.544   & 0.538   &0.54      &0.524   &0.501   &0.524  &0.537\\
\hline
$F(1.96)$  &0.9530  &0.9558   & 0.9545   &0.9623  &0.961    & 0.965   &0.971   &0.9842   &0.966  &0.913\\
\hline
\end{tabular}
\caption{}
	\label{tab:tab1}
\end{table}

\begin{table}[htbp]
{\small \centering
\begin{tabular}{|l|l|l|l|l|}
\hline
& $T^{*}_{n}(1/2)$ & Hill Est. & Pickands Est. & Lo Est. \\ \hline
Estimated $\gamma$ & 0.5047 & 0.5008 & 0,4992 & 0.4733 \\ \hline
MSE & 0.0051 & 0.0025 & 0.045 & 0.009 \\ \hline
\end{tabular}
}
\caption{}
\label{tab:tab2}
\end{table}

\begin{figure}[htbp]
	\centering
	\caption{}
		\includegraphics[width=10cm]{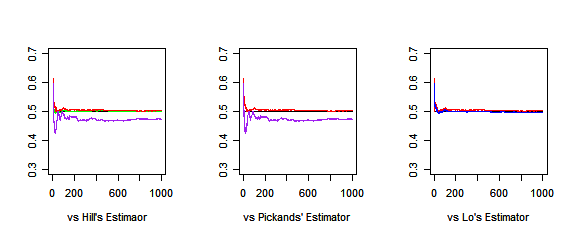}
		\label{fig:fig2}
\end{figure}

\begin{figure}[htbp]
	\centering
	\caption{}
		\includegraphics[width=10cm]{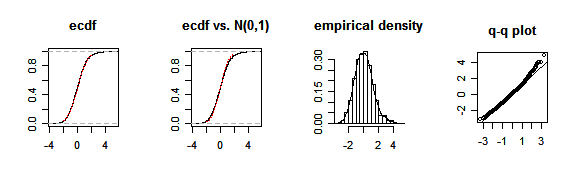}
	\label{fig:fig3}
\end{figure}

\begin{figure}[htbp]
	\centering
	\caption{}
		\includegraphics[width=10cm, height=5cm]{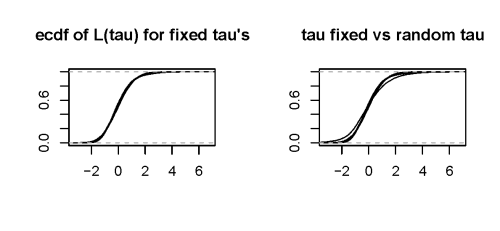}
	\label{fig:FIG. 1}
\end{figure}

\begin{figure}[htbp]
	\centering
	  \caption{}
		\includegraphics[width=10cm]{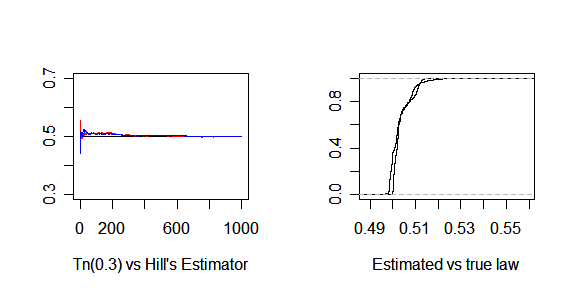}
	\label{fig:fig4}
\end{figure}

\end{document}